\documentclass[reprint,aps,prx,superscriptaddress,nofootinbib]{revtex4-2}
\usepackage[T1]{fontenc}
\usepackage{graphicx}
\usepackage{amsmath,amssymb}
\usepackage[colorlinks=true,linkcolor=blue,citecolor=blue,urlcolor=blue]{hyperref}
\usepackage{dsfont}
\usepackage{xcolor}

\newtheorem{theorem}{Theorem}
\newtheorem{corollary}{Corollary}
\newtheorem{reptheorem}{Theorem}
\newenvironment{theoremrestate}[1]{%
  \renewcommand{\thereptheorem}{\ref{#1}}%
  \begin{reptheorem}
}{%
  \end{reptheorem}%
}

\begin{document}
\title{Classical shadows for non-i.i.d. quantum sources}
\author{Leonardo Zambrano}
\email{leonardo.zambrano@icfo.eu}
\affiliation{ICFO - Institut de Ci\`encies Fot\`oniques, The Barcelona Institute of Science and Technology, 08860 Castelldefels, Barcelona, Spain}

\begin{abstract}
Classical shadows allow many properties of a quantum system to be estimated from relatively few measurements. Their standard guarantees, however, assume independent and identically distributed experimental rounds and a fixed analysis procedure. These assumptions can fail when the source drifts or the analysis adapts to previous measurement outcomes. We show that average expectation values along the run can be estimated with the usual classical-shadow sample-complexity scaling, even when the prepared states depend arbitrarily on the past and the analysis learns from earlier outcomes. For bounded shadow estimators, we further derive error bars that adapt to the observed data and remain valid throughout the experiment, allowing measurements to stop once the desired precision is reached. We illustrate these results numerically by certifying entanglement hidden by time averaging and reducing the measurements needed to estimate energy densities in many-body systems with small fluctuations.
\end{abstract}
\maketitle

\section{Introduction}

\begin{figure*}[t]
    \centering
    \includegraphics[width=0.96\textwidth]{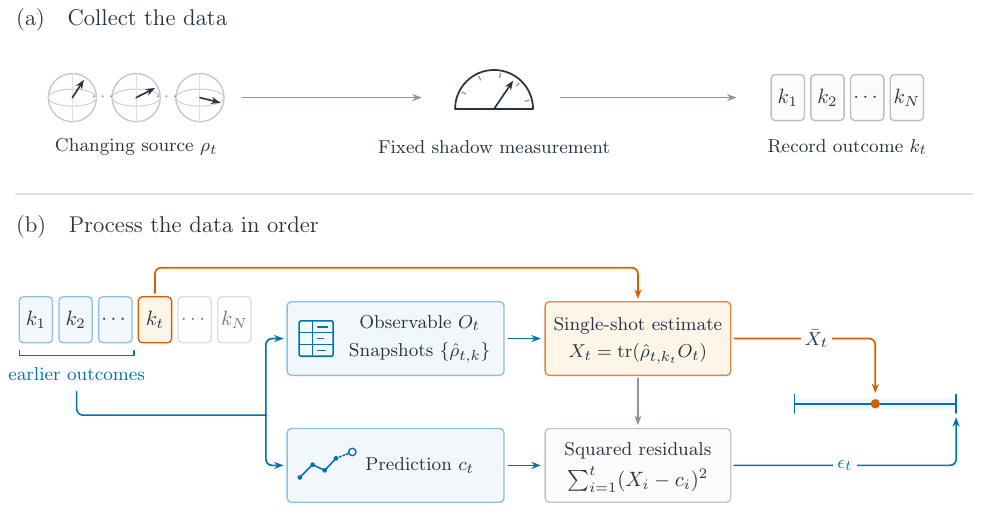}
    \caption{Adaptive classical-shadow estimation.
    (a) A fixed shadow measurement is applied to each preparation
    $\rho_t$, producing the recorded outcome $k_t$. For randomized
    measurements, $k_t$ includes both the setting and its outcome.
    (b) Using earlier outcomes, the experimenter chooses an
    observable, a dual frame, and a prediction $c_t$ of
    $\operatorname{tr}(\rho_t O_t)$. The current outcome selects
    the snapshot $\hat\rho_{t,k_t}$, giving the single-shot estimate
    $X_t=\operatorname{tr}(\hat\rho_{t,k_t} O_t)$.
    After processing round $t$, the sample mean $\bar X_t$
    estimates the trajectory average $\bar o_t$.
    The empirical confidence interval is centered at $\bar X_t$,
    with radius $\epsilon_t$ determined by the sum of squared
    residuals $\sum_{i=1}^{t}(X_i-c_i)^2$.
    Alternatively, the soft-truncated mean gives a fixed-sample
    guarantee controlled by the shadow norm, without predictions
    or squared residuals.
    \label{fig:shadow-protocol}}
\end{figure*}

Predicting properties of an unknown quantum system from a limited number
of measurements is a central task in quantum information science. Over the past few years, \emph{classical shadows} have become a standard framework for this purpose, enabling the efficient estimation
of many observables from a relatively small number of experimental
samples~\cite{huang2020predicting, nguyen2022optimizing,
hadfield2022measurements, innocenti2023shadow, caprotti2024optimizing, elben2023randomized}. By exploiting randomized
measurements together with efficient classical post-processing,
classical shadows circumvent the exponential overhead traditionally
associated with full quantum state tomography~\cite{hradil1997quantum,
haah2017sample, guta2020fast, surawy2022projected, zambrano2026fast}. This efficiency has enabled
applications ranging from quantum simulation and verification to the
detection of entanglement and the classification of quantum phases of
matter~\cite{zhao2021fermionic, struchalin2021experimental, huggins2022unbiasing, elben2020mixed, neven2021symmetry, huang2022provably}.

The conventional analysis of classical shadows assumes independent and
identically distributed (i.i.d.) rounds, each preparing the same fixed
quantum state. However, experimental platforms are subject to slow
parameter drift, residual environmental memory, and active feedback
that can change the prepared state or correlate measurement
rounds~\cite{klimov2018fluctuations,proctor2020detecting,
white2020demonstration,mcewen2021removing}. The state prepared in a
given round may therefore depend on the entire experimental history, and applying an i.i.d.\ error estimate to such a record can give a misleading impression of precision.

A second source of adaptation arises in the analysis of the data.
Shadow estimates can become more efficient when the reconstruction is tailored to the measured state, for example by optimizing the dual frame that defines the shadow snapshots~\cite{fischer2024dual, malmi2024enhanced,korhonen2025improving}.
One may also wish to follow a drifting source, or simply to stop
measuring as soon as the error bars are small enough. These choices can make the analysis depend on earlier
outcomes, so the usual statistical guarantees need not apply even
when the source itself is i.i.d.~\cite{huang2026predicting}. For instance, fitting a dual
frame on the same data used for estimation can introduce
bias~\cite{malmi2024enhanced,korhonen2025improving}, and a confidence
interval computed for a fixed sample size need not remain valid when the
stopping time is chosen from the
data~\cite{howard2021timeuniform,ramdas2023gametheoretic}.

This raises two questions: \emph{does the efficiency of
classical shadows survive beyond the i.i.d.\ regime, and can the analysis adapt to the data while retaining reliable error bars?}
Answering them is essential not only for certifying devices in drifting regimes, but also for extracting as much information as possible from each experiment.

In this work, we answer both questions in the affirmative. We allow the state $\rho_t$ prepared in round $t$ to depend arbitrarily on the past, and the experimenter to use earlier outcomes to choose both the observable $O_t$ and the single-shot estimator of its expectation value.  In this setting, there is generally no single fixed state to learn. We instead estimate the average expectation value along the realized trajectory, 
\begin{align}
    \bar{o}_N=\frac{1}{N}\sum_{t=1}^N\operatorname{tr}(\rho_t O_t).
\end{align}
Allowing the observable to change lets us
follow the source as it evolves, for instance by adjusting an
entanglement witness to a drifting preparation.

Using martingale concentration inequalities~\cite{williams1991probability}
together with a soft-truncated mean~\cite{catoni2012challenging,bhatt2022offline,wang2023catoni},
we show that estimating $\bar o_N$ with a sample size fixed in advance
retains the usual i.i.d.\ sample-complexity scaling governed by the
shadow norm, even when the source and the analysis depend on earlier
measurement outcomes.

The shadow norm provides a worst-case guarantee over all states,
but the estimates obtained in a particular experiment may exhibit
much smaller fluctuations~\cite{hayashi2026finite}. We therefore derive empirical confidence
intervals whose widths adapt to the observed
data~\cite{howard2021timeuniform,waudby2024estimating,
ramdas2023gametheoretic, chugg2025closed}. These intervals remain valid simultaneously
at all sample sizes, allowing the experiment to stop once the desired
precision is reached while preserving the confidence level.

The width of these intervals depends both on the estimators and on how well we predict the expectation value $\operatorname{tr}(\rho_t O_t)$ before each measurement. As illustrated in Fig.~\ref{fig:shadow-protocol}, refining the estimators and predictions using earlier outcomes can reduce the error bars while preserving the confidence guarantees. We can do this during acquisition or afterward by working through the recorded data in order.

We illustrate these results with two applications. First, we consider
an entangled source whose phase drifts during the experiment, so that
the time-averaged state  $\bar\rho_N=\frac{1}{N}\sum_{t=1}^N\rho_t$ is separable. A fixed witness cannot detect
entanglement in this average, whereas a witness that follows the drift
certifies that most individual preparations were entangled. Second, for many-body states with small energy fluctuations,
the variance of the single-shot estimators can be much smaller than
the shadow norm suggests. With sufficiently accurate predictions, the empirical bound reaches energy-density precision $\epsilon$ for an $n$-qubit system after $\mathcal{O}(1/(n\epsilon^2)+1/\epsilon)$ measurement rounds, up to logarithmic factors.

Our work is connected to several lines of research. On the source side, quantum de Finetti reductions
extend i.i.d.\ learning protocols to correlated inputs, typically at the
cost of additional samples~\cite{fawzi2024learning}. Classical-shadow protocols also accommodate independent, non-identical preparations, including drifting sources~\cite{neven2021symmetry}.
Martingale methods have been used to estimate the average of a fixed
entanglement witness under correlated noise~\cite{dirkse2020witnessing}, to verify and spot-check non-i.i.d.\ sources~\cite{navarro2026certifying,zhang2026spotchecking}, and to reconstruct the time-averaged state via full tomography~\cite{zambrano2026quantum}. For classical shadows, we show that trajectory averages retain the usual i.i.d.\ sample-complexity scaling under the same uniform
shadow-norm bound.  On the analysis side, optimized dual frames can reduce the variance of shadow estimators. However, these methods assume a stationary source
and fit the dual either on a separate calibration set or on the
estimation data itself without finite-sample
guarantees~\cite{fischer2024dual,malmi2024enhanced,korhonen2025improving}.
Empirical-variance stopping has been studied for energy estimation, again
with i.i.d.\ preparations~\cite{gresch2025reducing}. We instead adapt the estimator using earlier outcomes only. Every round then serves both learning and estimation, with error bars valid under arbitrary history dependence and
data-dependent stopping.

The paper is organized as follows. Section~\ref{sec:model} introduces
the measurement model and the history-dependent setting.
Section~\ref{sec:bounds} proves the fixed-sample guarantee and applies
it to entanglement certification along a drifting trajectory.
Section~\ref{sec:learning} develops the empirical confidence sequences,
the online learning of estimators and predictions, and
data-dependent stopping. Section~\ref{sec:applications} analyzes
extensive many-body observables, both theoretically and numerically.
We conclude in Sec.~\ref{sec:discussion}. Proofs, numerical details,
and an extension to quantum channels are given in the appendices.

\section{Measurement and reconstruction framework}
\label{sec:model}

Classical shadows rely on constructing unbiased single-shot estimators of expectation values from randomized measurements. Here we briefly review this static reconstruction mechanism, extend it to sequential experiments where both the prepared states and the estimators depend on past outcomes, and derive the conditional moments that govern our concentration bounds.

\subsection{Snapshots and single-shot estimators}

We consider a quantum system on a $d$-dimensional Hilbert space $\mathcal{H}$. A measurement is described by a positive operator-valued measure (POVM) $\{E_k\}_{k=1}^M$, where the effects satisfy $E_k\geq0$ and $\sum_kE_k=\mathds{1}$. Measuring a state $\rho$ gives outcome $k$
with probability
\begin{align}
 p_k(\rho)=\operatorname{tr}(\rho E_k).
\end{align}
For randomized measurements, $k$ labels both the measurement
setting and its outcome. The setting probability is included
in $E_k$, so a single POVM describes the full measurement protocol~\cite{oszmaniec2017simulating}.

To estimate the expectation value of an observable $O$, we assign a reconstruction coefficient $x_k \in \mathbb{R}$ to each possible outcome. We require these coefficients to satisfy
the reconstruction identity
\begin{align}
 \sum_k x_kE_k=O.
 \label{eq:reconstruction}
\end{align}
The coefficients define a single-shot estimator $X$, which takes
the value $x_k$ when outcome $k$ is observed. Its expectation is
\begin{align}
 \mathbb{E}[X]
 =\sum_k p_k(\rho)x_k
 =\operatorname{tr}(\rho O).
\end{align}
Thus any coefficients satisfying Eq.~\eqref{eq:reconstruction}
define an unbiased estimator of $\operatorname{tr}(\rho O)$.

The reconstruction identity can be satisfied whenever $O$ lies
in the span of the measurement effects. A POVM is informationally
complete (IC) if its effects span the space of Hermitian operators
on $\mathcal{H}$, so that every observable can be reconstructed.
If the effects are also linearly dependent, the POVM is
overcomplete and the reconstruction  $\{x_k\}$ is not unique.
Different choices give estimators with the same expectation
but potentially different variances. We will use this freedom to improve the
reconstruction using information from earlier measurement outcomes.

Classical shadows assign a Hermitian operator
$\hat\rho_k$, called a snapshot, to each measurement outcome
$k$~\cite{huang2020predicting}. These operators satisfy
\begin{align}
 \sum_k \operatorname{tr}(A E_k)\hat\rho_k = A
 \label{eq:dual-state}
\end{align}
for every Hermitian operator $A$ and therefore form a dual
frame of the measurement effects~\cite{innocenti2023shadow}.

Setting $A=\rho$ gives $\mathbb{E}_{k\sim p(\rho)}[\hat\rho_k]=\rho$,
so the snapshot is an unbiased estimator of the state. The same snapshot defines a single-shot estimator for any observable $O$ by
setting 
\begin{align}
 x_k=\operatorname{tr}(\hat\rho_k O).
\end{align}
By linearity, $X=x_k$ is an unbiased estimator of $\operatorname{tr}(\rho O)$
for every state $\rho$, so the reconstruction coefficients $x_k$ satisfy
Eq.~\eqref{eq:reconstruction}.

Two standard classical-shadow protocols use global Clifford
measurements and local Pauli measurements~\cite{huang2020predicting}.
For $n$ qubits, with $d=2^n$, the global Clifford protocol
applies a uniformly random $U\in\mathrm{Cl}(n)$ and then
measures in the computational basis. With the outcome labeled by $(U,z)$, the POVM effects and standard snapshots are
\begin{align}
 E_{U,z}
 &=\frac{1}{|\mathrm{Cl}(n)|}
 U^\dagger|z\rangle\langle z|U,
 \notag\\
 \hat\rho_{U,z}
 &=(d+1)U^\dagger|z\rangle\langle z|U-\mathds{1}.
 \label{eq:clifford-dual}
\end{align}
This protocol is useful for estimating fidelities with pure
target states and expectation values of observables with
small Hilbert--Schmidt norm.

For local Pauli measurements, each qubit is measured independently
in a uniformly random $X$, $Y$, or $Z$ basis. Let $U_j$ rotate
the chosen basis on qubit $j$ to the computational basis, and
let $z_j\in\{0,1\}$ denote its outcome. Writing
$U=(U_1,\ldots,U_n)$ and $z=(z_1,\ldots,z_n)$, we have
\begin{align}
 E_{U,z}
 &=\frac{1}{3^n}\bigotimes_{j=1}^n
 U_j^\dagger|z_j\rangle\langle z_j|U_j,
 \notag\\
 \hat\rho_{U,z}
 &=\bigotimes_{j=1}^n
 \left(3U_j^\dagger|z_j\rangle\langle z_j|U_j-\mathds{1}_j\right).
 \label{eq:pauli-dual}
\end{align}
This protocol is well suited to observables built from
low-weight Pauli strings, where the weight is the number
of nonidentity factors. Examples include few-body correlators
and local Hamiltonians.

\subsection{History dependence and adaptation}
\label{sec:martingales}

We now consider a repeated experiment in which both the
prepared states and the choice of observables and estimators
may depend on the past.
Our goal is to estimate the trajectory average,
\begin{align}
 \bar o_N=\frac{1}{N}\sum_{t=1}^N
 \operatorname{tr}(\rho_t O_t),
 \label{eq:adaptive-trajectory-target}
\end{align}
where $\rho_t$ is the state prepared in round $t$ and $O_t$
is the observable assigned to that round.
The observables and estimators may be chosen during the
experiment or afterward by processing the recorded data
in order. In either case, the choices for round $t$ use
only information from earlier rounds. The resulting
trajectory average describes that particular run and
may vary between runs.

Let $\mathcal F_{t-1}$ denote the full experimental history
before the measurement in round $t$. It includes earlier
measurement settings and outcomes, together with source
variables such as drift parameters or feedback signals,
some of which may be inaccessible to the experimenter.
These histories form an increasing sequence $\mathcal F_0\subseteq\mathcal F_1\subseteq\cdots$, called a filtration.

We model $\rho_t$ as a random density operator determined
by this history. The full history fixes the state in each
run, while different histories can lead to different states.
Mathematically, $\rho_t$ is $\mathcal F_{t-1}$-measurable,
and the state sequence is called predictable.
Here, predictable does not mean that the experimenter
knows the state or can forecast it accurately, since
some variables in the full history may be inaccessible.

For each $t$, we measure $\rho_t$ with the fixed POVM $\{E_k\}$
and record the outcome $k_t$, as illustrated in
Fig.~\ref{fig:shadow-protocol}(a).
For randomized measurements, the current setting is drawn
independently of $\mathcal F_{t-1}$, and hence independently
of the prepared state.
The recorded outcome $k_t$ then has conditional distribution
\begin{align}
 \mathbb P(k_t=k\mid\mathcal F_{t-1})
 =\operatorname{tr}(\rho_t E_k ).
 \label{eq:conditional-born}
\end{align}

To analyze round $t$, the experimenter uses only the accessible
past information to choose the observable $O_t$ and real
reconstruction coefficients $\{x_{t,k}\}_k$ satisfying
\begin{align}
 \sum_k x_{t,k}E_k=O_t.
 \label{eq:predictable-reconstruction}
\end{align}
Both the observable and its reconstruction are therefore
predictable.
The recorded outcome $k_t$ gives the single-shot estimate
\begin{align}
 X_t=x_{t,k_t}.
\end{align}
For classical shadows, this becomes
$X_t=\operatorname{tr}(\hat\rho_{t,k_t} O_t)$,
where the snapshots $\{\hat\rho_{t,k}\}_k$ form a dual frame of $\{E_k\}_k$
chosen using earlier data.

Fig.~\ref{fig:shadow-protocol}(b) illustrates how the
recorded outcomes are processed. When processing round $t$:
\begin{enumerate}
\item Use the accessible information from earlier rounds
to choose $O_t$ and $\{\hat\rho_{t,k}\}_k$.
\item Use the recorded outcome $k_t$ to select the snapshot
$\hat\rho_{t,k_t}$ and compute the estimate
$X_t=\operatorname{tr}(\hat\rho_{t,k_t} O_t)$.
\item Include the data from round $t$ in the history used
to make subsequent choices.
\end{enumerate}
Fig.~\ref{fig:shadow-protocol} also includes a prediction
$c_t$ of $\operatorname{tr}(\rho_t O_t)$, chosen using earlier
data. In Sec.~\ref{sec:learning}, we use the sum of squared residuals
$\sum_{i=1}^{t} (X_i-c_i)^2$ to construct empirical confidence intervals whose widths adapt to the observed data.

\subsection{Conditional moments of single-shot estimators}

Conditional on $\mathcal{F}_{t-1}$, both the reconstruction
coefficients $\{x_{t,k}\}_k$ and the prepared state are fixed. Using
Eqs.~\eqref{eq:conditional-born}
and~\eqref{eq:predictable-reconstruction}, we obtain
\begin{align}
 \mathbb{E}[X_t\mid\mathcal{F}_{t-1}]
 =\sum_k x_{t,k}\operatorname{tr}(\rho_t E_k )=\operatorname{tr}(\rho_t O_t).
 \label{eq:martingale-principle}
\end{align}
Thus the single-shot estimator remains conditionally
unbiased even when the source, observable, and
reconstruction coefficients depend on earlier outcomes.
The errors $X_t-\operatorname{tr}(\rho_t O_t)$ therefore
form a martingale difference sequence, allowing us to
apply concentration inequalities that remain valid
for dependent measurement rounds.

Different reconstructions $\{x_{t,k}\}_k$ can nevertheless give different
fluctuations. Their conditional second moment satisfies
\begin{align}
 \mathbb{E}[X_t^2 \! \mid \! \mathcal{F}_{t-1}] =\operatorname{tr}\left(\! \rho_t\sum_k x_{t,k}^2E_k \! \right) \!\leq \! \left\lVert\sum_k x_{t,k}^2E_k\right\rVert_\infty \!\!,
 \label{eq:conditional-second-moment}
\end{align}
where $\lVert\cdot\rVert_\infty$ denotes the operator norm.
For a fixed observable and shadow reconstruction, this bound
is the squared shadow norm $\lVert O\rVert_{\mathrm S}^2$~\cite{huang2020predicting,caprotti2024optimizing}.
More generally, we require only a constant $s^2$ that bounds
the conditional second moment in every round.
Bounds for standard reconstructions are given in
App.~\ref{app:shadow-moments}.

Since all exact reconstructions of $O_t$ have the same
conditional mean $\operatorname{tr}(\rho_t O_t)$, reducing the
conditional second moment also reduces the conditional
variance,
\begin{align}
 \operatorname{Var}(X_t\mid\mathcal{F}_{t-1})
 =\mathbb{E}[X_t^2\mid\mathcal{F}_{t-1}]
 -[\operatorname{tr}(\rho_t O_t)]^2.
\end{align}
We will use this freedom to learn reconstructions
that give more precise estimates.

\section{Classical-shadow guarantees for history-dependent sources}
\label{sec:bounds}

We first show that the trajectory average introduced in
Sec.~\ref{sec:model} can be estimated with the usual
shadow-norm scaling, even when the preparations,
observables, and reconstructions depend on earlier outcomes. We then use
an adaptive entanglement witness to certify entanglement
along a trajectory whose time-averaged state is separable.

\subsection{Estimating trajectory averages}
\label{sec:catoni}

We estimate $\bar{o}_N=\frac{1}{N}\sum_{t=1}^N \operatorname{tr}(\rho_t O_t)$ from the single-shot estimators $X_t$. To obtain a confidence bound controlled by
their conditional second moments, we use the soft-truncated
mean~\cite{catoni2012challenging,bhatt2022offline,wang2023catoni}
\begin{align}
 \bar{Y}_N =\frac{1}{\alpha N}\sum_{t=1}^N\psi(\alpha X_t), \label{eq:catoni-estimator}
\end{align}
where $\alpha>0$ and 
\begin{align}
 \psi(u)
 =\operatorname{sgn}(u)
 \log\left(1+|u|+\frac{u^2}{2}\right).
 \label{eq:catoni-function}
\end{align}
This function is approximately linear near the origin
and grows logarithmically for large arguments. The soft-truncated mean is therefore close to the sample mean when all
$|\alpha X_t|$ are small, while reducing the influence of
unusually large single-shot estimates.

The following theorem bounds the estimation error using a uniform
bound on the conditional second moments~\cite{bhatt2022offline}.
\begin{theorem}
\label{thm:catoni}
Suppose that, for every $t$,
\begin{align}
 \mathbb{E}[X_t\mid\mathcal{F}_{t-1}]=\operatorname{tr}(\rho_t O_t),
 \qquad
 \mathbb{E}[X_t^2\mid\mathcal{F}_{t-1}]\leq s^2,
\end{align}
with a fixed constant $s>0$.
For a sample size $N$ fixed in advance and
$\delta\in(0,1)$, choose
$\alpha=\sqrt{2\log(2/\delta)/(Ns^2)}$.
Then, with probability at least $1-\delta$,
\begin{align}
 \left|\bar{Y}_N -\bar o_N\right|
 \leq \sqrt{\frac{2s^2\log(2/\delta)}{N}}.
 \label{eq:catoni-bound}
\end{align}
\end{theorem}

The proof is given in App.~\ref{app:catoni}.
The bound retains the familiar shadow-norm scaling of
the i.i.d.\ setting, even when the source, observables,
and reconstruction depend on earlier outcomes.

The same measurement data can be used to estimate $K$
trajectory averages simultaneously. Suppose the single-shot estimators
for all $K$ targets satisfy the same conditional
second-moment bound $s^2$.
Choosing $\alpha=\sqrt{2\log(2K/\delta)/(Ns^2)}$ for each target, we obtain
\begin{align}
 \left|\bar{Y}_N^{(j)}-\bar o_N^{(j)}\right|
 \leq\epsilon
 \quad\text{for all }j=1,\ldots,K,
\end{align}
with probability at least $1-\delta$, provided the sample
size is fixed in advance and satisfies
\begin{align}
 N\geq
 \frac{2s^2}{\epsilon^2}\log\frac{2K}{\delta}.
 \label{eq:catoni-samples}
\end{align}
Thus the logarithmic dependence on the number of targets is preserved in this setting.

A known identity component need not be estimated: we apply Theorem~\ref{thm:catoni} to the
traceless part of $O_t$ and add the known average of $\operatorname{tr}(O_t)/d$ to the
resulting soft-truncated mean.

\subsection{Certifying entanglement along a trajectory}
\label{sec:witness}

 Phase drift can hide entanglement when preparations are averaged over time. Consider the two-qubit source
\begin{align}
 \rho_t
 &=v|\Phi_{\phi_t}\rangle\langle\Phi_{\phi_t}|
   +(1-v)\frac{\mathds{1}}{4}, \qquad 0\leq v\leq1,
 \label{eq:bell-source}
\end{align}
where $|\Phi_\phi\rangle=\frac{1}{\sqrt{2}}
(|00\rangle+e^{i\phi}|11\rangle)$.
For $v>\frac{1}{3}$, each state is entangled, while a complete
phase rotation $\phi_t=\frac{2\pi}{N}(t-1)$ gives
a separable time-averaged state $\bar\rho_N=\frac{1}{N}\sum_{t=1}^N\rho_t$ because the off-diagonal
terms cancel.

To follow the drift, we use earlier outcomes to choose
the phase $\theta_t$ of the witness
$W_t=\mathds{1}-2|\Phi_{\theta_t}\rangle
\langle\Phi_{\theta_t}|$~\cite{guhne2009entanglement}.
Each qubit is then measured independently in a uniformly
random Pauli basis. The resulting shadow snapshot provides
a single-shot estimate of the witness expectation, and
the new data are used to choose the next witness.

For a separable state, $-\operatorname{tr}(\rho_t W_t)\leq0$,
while for an arbitrary state, $-\operatorname{tr}(\rho_t W_t)\leq1$.
The average negative witness expectation $-\frac{1}{N}\sum_{t=1}^{N} \operatorname{tr}(\rho_t W_t)$ therefore cannot
exceed the fraction of entangled preparations.
If $L_N$ is any lower confidence bound on this average, then
\begin{align}
 f_N=\frac{1}{N}\sum_{t=1}^N
 \mathbf1\{\rho_t\text{ is entangled}\}
 \geq\max\{0,L_N\}.
 \label{eq:entangled-fraction}
\end{align}

At a sample size fixed in advance, we take $O_t=-W_t$, whose identity component is
$-\operatorname{tr}(W_t)/4=-1/2$. Applying Theorem~\ref{thm:catoni} to the traceless part
$-W_t+\mathds{1}/2$ and adding $-\frac{1}{2}$ gives the fixed-sample lower bound
$L_N=\bar Y_N-\frac{1}{2}-\epsilon$, where $\epsilon$ is the radius from the theorem. Eq.~\eqref{eq:entangled-fraction} then holds with
probability at least $1-\delta$.
Further details are given in
Appendices~\ref{app:witness} and~\ref{app:numerics}.

For $v=0.9$ and a complete phase rotation over
$N=32768$ measurements, we compare the adaptive witness
with the fixed witness $W_0$, using the same measurement
outcomes.
Both use the standard local Pauli reconstruction from Eq.~\eqref{eq:pauli-dual} and the soft-truncated
mean at $N$, with $\delta=0.05$ for each certificate.
Across $50$ independent repetitions, the adaptive witness
gives a median lower bound of $0.797$ on the entangled
fraction. The fixed witness gives no positive certificate
in any repetition.
Thus the adaptive witness certifies entanglement along
the trajectory even though the time-averaged state $\bar\rho_N$
is separable.

\section{Learning to obtain smaller error bars}
\label{sec:learning}

 The guarantee of Sec.~\ref{sec:bounds} uses a second-moment
bound that holds uniformly over possible states.
The fluctuations in a particular experiment can be much
smaller. We now construct empirical confidence intervals that adapt
to the observed fluctuations and remain valid throughout
data collection, allowing measurements to stop once
a desired precision is reached.

\subsection{Confidence and learning from observed fluctuations}
\label{sec:empirical}
\label{sec:learning-rules}

We estimate the trajectory average $\bar o_N$ using the
sample mean $\bar X_N=\frac{1}{N}\sum_{t=1}^N X_t$.
When processing round $t$, we first use only earlier
outcomes to choose a prediction $c_t$ of $\operatorname{tr}(\rho_t O_t)$.
We then compute $X_t$, form the residual $X_t-c_t$,
and update the sum of squared residuals,
\begin{align}
 S_N=\sum_{t=1}^N(X_t-c_t)^2.
 \label{eq:squared-residuals}
\end{align}

Conditional unbiasedness gives
\begin{align}
\!\!\!\mathbb{E}[(X_t-c_t)^2\!\mid\!\mathcal{F}_{t-1}]
 \!=\!\!\operatorname{Var}(X_t\!\mid\!\mathcal{F}_{t-1})\!+\!(\operatorname{tr}(\rho_t O_t)\!-\!c_t)^2\!\!.\!
 \label{eq:variance-prediction}
\end{align}
There are therefore two ways to reduce the expected
squared residuals: choose a single-shot estimator with lower
conditional variance, or predict the changing expectation
value more accurately.

The following theorem gives an empirical confidence radius $\epsilon_N$ based on
$S_N$~\cite{howard2021timeuniform}. The intervals $[\bar X_N-\epsilon_N,\bar X_N+\epsilon_N]$ form an empirical confidence sequence: they cover the trajectory average simultaneously at all sample sizes.

\begin{theorem}
\label{thm:eb}
Let $c_t$ be predictable and suppose that, for every $t$,
\begin{align}
\mathbb{E}[X_t\mid\mathcal{F}_{t-1}]=\operatorname{tr}(\rho_t O_t),
 \qquad |X_t-c_t|\leq B,
\end{align}
with a fixed constant $B>0$.
Let $\phi(\lambda)=-\log(1-\lambda)-\lambda$.
For $\delta\in(0,1)$, define
\begin{align}
\epsilon_N
 =\min_{0\leq j\leq J_N}
 \frac{B^2\log (2/\delta a_j)
       +\phi(\lambda_j)S_N}
      {N\lambda_j B},
 \label{eq:empirical-radius}
\end{align}
where $a_j=\frac{1}{(j+1)(j+2)}$,
$\lambda_j=\frac{1}{2^{(j+1)}}$, and
$J_N=\lceil\frac12\log_2N\rceil$.
Then, with probability at least $1-\delta$,
\begin{align}
 |\bar X_N-\bar o_N|\leq\epsilon_N
 \label{eq:empirical-coverage}
\end{align}
for every $N \geq 1$. 
\end{theorem}

The proof and the bound
\begin{align}
 \epsilon_N
 \leq
 \frac{3\sqrt{S_N\ell_N}}{N}
 +\frac{4B\ell_N}{N}
 \label{eq:empirical-width}
\end{align}
are given in App.~\ref{app:eb}, where
$\ell_N=\log[2(J_N+1)(J_N+2)/\delta]$.
The first term depends on the observed residuals through
$S_N$, while the second depends on the fixed residual
bound $B$.
When the observed residuals are small, this radius can
be smaller than the fixed-sample radius of Theorem~\ref{thm:catoni}.

The decomposition in Eq.~\eqref{eq:variance-prediction}
suggests how to use earlier data to obtain a smaller
empirical radius $\epsilon_N$. For an overcomplete POVM, different reconstruction coefficients can
estimate the same observable with different variances
\cite{innocenti2023shadow,caprotti2024optimizing,fischer2024dual,malmi2024enhanced,
korhonen2025improving}. 
Let $\sigma_t \geq 0$ be a state estimate constructed
from earlier outcomes, and set
$\hat p_{t,k}=\operatorname{tr}(\sigma_t E_k)$.
Before round $t$, we choose the reconstruction coefficients for $O_t$
by solving~\cite{innocenti2023shadow,malmi2024enhanced}
\begin{align}
 \min_x\quad&
 \sum_k\hat p_{t,k}x_k^2,
 \nonumber\\
 \text{subject to}\quad&
 \sum_kx_kE_k=O_t,
 \nonumber\\
 &|x_k-c_t|\leq B
 \qquad\text{for every }k.
 \label{eq:online-dual}
\end{align}
All feasible reconstructions have the same mean under
$\sigma_t$, so minimizing their second moment also minimizes
their variance. The first constraint ensures conditional unbiasedness
for the actual state $\rho_t$, even if $\sigma_t$
is inaccurate. The second ensures the residual bound
required by Theorem~\ref{thm:eb}. An inaccurate state estimate can therefore reduce
efficiency without invalidating the guarantee. Details are given in App.~\ref{app:learning}.

We can predict the expectation value from the state estimate,
$c_t=\operatorname{tr}(\sigma_tO_t)$.
For a fixed observable, an alternative is
$c_{t+1}=(1-\gamma)c_t+\gamma X_t$, followed by clipping
to a fixed interval $I$. The weight $0<\gamma\leq1$
controls the contribution of the latest estimate.
We choose $I$ and $B$ to keep Eq.~\eqref{eq:online-dual}
feasible for every $c_t\in I$.
A larger $B$ allows more reconstructions but increases
the term proportional to $B$ in Eq.~\eqref{eq:empirical-width}.

We return to the drifting Bell source to compare the
empirical confidence bound using learned reconstruction
with the fixed-sample benchmark of Sec.~\ref{sec:witness}. We use $B=9$ and $\delta=0.05$.
After $N=32768$ measurements, the median empirical lower bound
on the entangled fraction over $50$ repetitions is
$0.806$, compared with $0.797$ for the fixed-sample bound computed from the same data.

\subsection{Data-dependent stopping}
\label{sec:stopping}

 Theorem~\ref{thm:eb} allows measurements to stop once
a desired precision or certificate is reached.
For a prescribed precision $\epsilon>0$, we stop when
the empirical radius first falls to $\epsilon$ or below:
\begin{align}
 \tau_\epsilon
 =\inf\{N\geq1:\epsilon_N\leq\epsilon\}.
 \label{eq:stopping-rule}
\end{align}
Since the theorem holds simultaneously at every
sample size, this choice preserves the confidence level.

\begin{corollary}
\label{cor:stopping}
Under the assumptions of Theorem~\ref{thm:eb}, if measurements
continue until the threshold in Eq.~\eqref{eq:stopping-rule}
is reached, then 
\begin{align}
 \mathbb{P}\left(
 |\bar X_{\tau_\epsilon}-\bar o_{\tau_\epsilon}|
 \leq\epsilon
 \right)\geq1-\delta,
 \label{eq:stopped-precision}
\end{align}
where $\bar o_{\tau_\epsilon}$ is the trajectory average up to stopping.
\end{corollary}

For entanglement certification, applying Theorem~\ref{thm:eb} to $O_t=-W_t$ gives the
empirical lower bound $L_N=\bar X_N-\epsilon_N$,
valid simultaneously for all $N$.
We can therefore stop as soon as $L_N\geq f_\star$.
Eq.~\eqref{eq:entangled-fraction} then certifies,
with confidence at least $1-\delta$, that at least
a fraction $f_\star$ of the preparations measured
up to stopping were entangled.
The threshold may never be reached if the source
or the chosen witnesses do not support the certificate.

For the Bell source, we choose the threshold
$f_\star=0.80$.
With learned reconstruction, all $50$ repetitions reach
this threshold within $65536$ measurements, with a median
stopping time of $19268$ and an interquartile range
of $13751$--$28365$. The phase period is fixed at $32768$ measurements, regardless of when
the measurements stop.

\section{Improved precision for many-body observables}
\label{sec:applications}
\label{sec:manybody}

Bounds based on second moments can overestimate the
fluctuations of randomized measurement estimators around
their mean. The total energy, for example, can have a mean
and a physical variance both of order $n$.
We show that standard local Pauli estimators can likewise
have variance of order $n$, while their second moment
is of order $n^2$, dominated by the squared mean.
The empirical bound of Sec.~\ref{sec:learning}
can exploit this smaller variance when predictions
track the mean. We then examine numerically the resulting
gains in precision and measurement budget using the standard
local Pauli reconstruction and predictions learned from
earlier outcomes.

\subsection{Small fluctuations in extensive observables}

\begin{figure*}[t]
 \centering
 \includegraphics[width=0.92\textwidth]{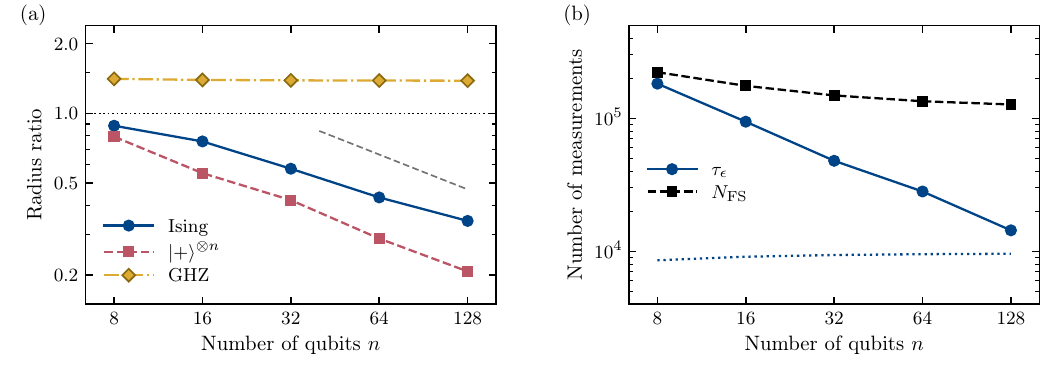}
 \caption{
Precision and measurement budgets for many-body observables.
(a)~Ratio of the empirical radius $r_N$ to the fixed-sample
radius $r_N^{\mathrm{FS}}$ after $N=2^{16}$ uniform local
Pauli measurements.
Circles show the energy density of the open transverse-field Ising chain in its ground state.
Squares and diamonds show the magnetization density of
$|+\rangle^{\otimes n}$ and the GHZ state, respectively.
The dotted horizontal line marks a ratio of one, and the
gray dashed segment shows the reference scaling $n^{-1/2}$.
(b)~Measurements needed to reach energy-density precision
$\epsilon=0.01$.
Circles show empirical stopping times, squares show the
deterministic fixed-sample budget $N_{\mathrm{FS}}$,
and the dotted line shows the lower bound
$2b\log(4/\delta)/\epsilon$ for the empirical stopping rule,
with $b=11-10/n$.
Both methods use standard local Pauli shadows.
Each method and target uses $\delta=0.05$.
Markers show medians over $50$ independent repetitions
for each state and system size.
Interquartile ranges are smaller than the markers.}
 \label{fig:mb}
\end{figure*}

Consider the total magnetization $H=\sum_{i=1}^n Z_i$.
For the standard local Pauli snapshot $\hat\rho_k$
in Eq.~\eqref{eq:pauli-dual}, the single-shot estimator
$X=\operatorname{tr}(\hat\rho_k H)$ satisfies
\begin{align}
 \operatorname{Var}_\rho(X)=\operatorname{Var}_\rho(H)+2n,
 \qquad
 \lVert H\rVert_{\mathrm S}^2=n^2+2n.
 \label{eq:magnetization}
\end{align}
Here $\operatorname{Var}_\rho(H)$ is the variance of a direct
measurement of the total magnetization, and the randomized
local Pauli measurements contribute an additional $2n$.
For product states, $\operatorname{Var}_\rho(H)\leq n$, so
the variance of the single-shot estimator grows only linearly
with system size, whereas its worst-case second moment
grows quadratically. When predictions track the expectation
value sufficiently well, the empirical confidence bound of
Theorem~\ref{thm:eb} can exploit these smaller fluctuations.

The same mechanism applies to local Hamiltonians
\begin{align}
 H=\sum_{j=1}^m\alpha_jP_j,
 \label{eq:local-hamiltonian}
\end{align}
where each Pauli string has weight at most $k$, each qubit
belongs to at most $\ell$ terms, and $|\alpha_j|\leq J$.
For the standard local Pauli single-shot estimator,
App.~\ref{app:manybody} shows that
\begin{align}
 \operatorname{Var}_\rho(X)
 \leq \operatorname{Var}_\rho(H)+3^k k\ell mJ^2.
 \label{eq:measurement-variance}
\end{align}
For fixed $k$, $\ell$, and $J$,
the additional measurement variance is $O(n)$.
Whenever the physical energy variance satisfies
$\operatorname{Var}_{\rho_t}(H)=O(n)$, this gives
$\operatorname{Var}(X_t\mid\mathcal{F}_{t-1})=O(n)$.
This condition holds for energy eigenstates and for states
with sufficiently weak spatial correlations.

We now estimate the average energy per qubit along
the trajectory using
\begin{align}
 \bar e_N&=\frac{1}{nN}\sum_{t=1}^N\operatorname{tr}(\rho_t H),
 &
 \hat e_N&=\frac{1}{nN}\sum_{t=1}^N X_t.
 \label{eq:energy-density-target}
\end{align}
Before each measurement, we use earlier outcomes to predict
the total energy, choosing $c_t\in[-mJ,mJ]$.
This interval contains every possible mean energy.
Since $|X_t|\leq 3^k mJ$, the residual satisfies
$|X_t-c_t|\leq B$, with $B=(3^k+1)mJ$.

Theorem~\ref{thm:eb} then guarantees
$|\hat e_N-\bar e_N|\leq r_N$ simultaneously for all $N$,
with probability at least $1-\delta$, where
\begin{align}
 r_N=\frac{\epsilon_N}{n},
 \qquad
 S_N=\sum_{t=1}^N(X_t-c_t)^2.
 \label{eq:density-empirical-radius}
\end{align}
Here $\epsilon_N$ is evaluated directly on the finite grid in
Eq.~\eqref{eq:empirical-radius}. Small variance and accurate
predictions help make this radius small, but are not required
for coverage.

To quantify how the radius depends on system size and prediction
accuracy, define the accumulated squared prediction error
\begin{align}
 D_N=\sum_{t=1}^N[\operatorname{tr}(\rho_t H)-c_t]^2.
 \label{eq:energy-prediction-error}
\end{align}
Using the variance--prediction decomposition in
Eq.~\eqref{eq:variance-prediction}, together with
$\operatorname{Var}(X_t\mid\mathcal F_{t-1})=O(n)$
and $B=O(n)$, App.~\ref{app:manybody} shows that,
with probability at least $1-\eta$, simultaneously for all $N$,
\begin{align}
 r_N
 =\mathcal{O}\left(
 \frac{1}{\sqrt{nN}}
 +\frac{\sqrt{D_N}}{nN}
 +\frac{1}{N}
 \right),
 \label{eq:density-prediction-scaling}
\end{align}
up to logarithmic factors.
Thus, with probability at least $1-\delta-\eta$,
the confidence intervals contain the average energy density
and have widths with the scaling above.

If $D_N=\mathcal{O}(nN+n^2)$ uniformly in $N$, this reduces to
\begin{align}
 r_N
 =\mathcal{O}\left(
 \frac{1}{\sqrt{nN}}+\frac{1}{N}
 \right).
 \label{eq:density-scaling}
\end{align}
This condition allows an average squared prediction error of
order $n$, together with a total initialization contribution
of order $n^2$. For example, errors of order $n$ during a bounded
number of initial rounds fit within this allowance.
The resulting bound shows how increasing system size can
improve energy-density precision: up to logarithmic factors,
the variance contribution decreases as $1/\sqrt{nN}$,
while the $1/N$ contribution remains.

\subsection{Numerical precision and stopping times}
\label{sec:numerical-simulations}

We compare the fixed-sample bound of
Theorem~\ref{thm:catoni} with the empirical bound of
Theorem~\ref{thm:eb} for estimating energy density.
We consider ground states of the open transverse-field
Ising chain,
\begin{align}
 H=-\sum_{i=1}^{n-1}Z_iZ_{i+1}-\sum_{i=1}^{n}X_i,
\end{align}
with $n=8,16,32,64,128$.
Both methods use standard local Pauli shadows and the same
measurement data, with $\delta=0.05$.
The empirical method uses predictions based on earlier outcomes.
We report medians over $50$ independent repetitions
for each state and system size.
Implementation details are given in
App.~\ref{app:mb-numerics}.

Fig.~\ref{fig:mb}(a) shows that after $N=2^{16}$ measurements,
the ratio of the empirical radius to the fixed-sample radius
decreases from $0.88$ at $n=8$ to $0.34$ at $n=128$.
The ground states have negligible physical energy variance,
so the fluctuations arise predominantly from the randomized
measurements. Good predictions help the empirical bound exploit these small
fluctuations: at $n=128$, the empirical radius fixing $c_t=0$ is
$3.95$ times the radius obtained with learned predictions,
using the same measurement data.

The magnetization density $\frac{1}{n} \sum_iZ_i$ illustrates
the role of physical fluctuations.
For $|+\rangle^{\otimes n}$, the radius ratio decreases
from $0.79$ to $0.21$ over the same sizes.
For the GHZ state, the ratio remains between $1.38$
and $1.41$. Like the product state, it has zero mean
magnetization, and measuring $Z$ on any individual qubit
gives outcomes $+1$ and $-1$ with equal probability.
Its total magnetization variance grows as $n^2$,
preventing the improvement observed for the product state.

Fig.~\ref{fig:mb}(b) shows the measurement budget needed
to reach an energy-density precision $\epsilon=0.01$.
The median stopping time decreases from
$1.82\times10^5$ measurements at $n=8$ to
$1.44\times10^4$ at $n=128$.
The latter is approximately $11\%$ of the fixed-sample
budget and only $1.5$ times the lower bound for this
empirical stopping rule.
This lower bound comes from a $1/N$ contribution to the
confidence bound that remains even when the observed
residuals vanish.  Increasing system size reduces the
contribution from fluctuations but does not shrink this
term, leaving limited room for further reductions
in the measurement budget.

\section{Discussion and conclusions}
\label{sec:discussion}

We have shown that classical shadows remain efficient when the prepared
states, observables, and estimators depend on earlier outcomes.
For a sample size fixed in advance, trajectory averages $\bar{o}_N = \frac{1}{N}\sum_t\operatorname{tr}(\rho_t O_t)$ can be estimated
with the usual i.i.d.\ sample-complexity scaling set by the shadow norm.
For bounded residuals, empirical confidence sequences give error bars
that adapt to the observed data and remain valid at every sample size.
The analysis can therefore learn from earlier measurements, and the
experiment can stop as soon as the desired precision is certified.

The principle that makes this adaptation possible is conditional
unbiasedness,
$\mathbb{E}[X_t\mid\mathcal{F}_{t-1}]=\operatorname{tr}(\rho_t O_t)$.
Any exact reconstruction $\{x_{t,k}\}_k$ chosen from earlier outcomes
preserves this property, so the error bars remain valid however the
reconstruction is chosen. This leaves room to improve precision:
learning the reconstruction can reduce the estimator variance, and
tracking the changing expectation value reduces the prediction error.
Both shrink the empirical error bars, while inaccurate learning only
makes them wider. Each measurement can thus contribute both to
estimation and to improving the analysis of later rounds.

The applications illustrate the physical consequences of adaptation.
A witness that follows a phase drift certifies a lower bound on the
fraction of entangled preparations, even when the time-averaged state $\bar{\rho}_N = \frac{1}{N}\sum_t \rho_t$
is separable. A fixed witness cannot detect this entanglement, since its trajectory average equals its expectation value in $\bar\rho_N$. For local many-body Hamiltonians, small
energy fluctuations and sufficiently accurate predictions instead reduce
the measurement budget. When the energy variance grows at most linearly
with the number of qubits $n$, the sufficient budget for energy-density
precision $\epsilon$ scales as $\mathcal{O}(1/(n\epsilon^2)+1/\epsilon)$,
up to logarithmic factors. 

The guarantees concern the trajectory average of the actual
preparations, including their noise and drift, rather than an ideal
target state. They assume a fixed, known measurement, which need
not be ideal. For noisy measurements, noise-adapted dual
frames~\cite{chen2021robust,koh2022classical} fit directly
into our framework.

\begin{acknowledgments}
The author is grateful to Mariana Navarro, Luciano Pereira and Antonio
Ac\'in for fruitful discussions. This work was supported by the Government
of Spain (Severo Ochoa CEX2019-000910-S, FUNQIP, and
NextGenerationEU PRTR-C17.I1), Fundaci\'o Cellex, Fundaci\'o Mir-Puig,
Generalitat de Catalunya (CERCA program) and the EU and Spanish AEI project QEC4QEA.
\end{acknowledgments}

\bibliography{bib}

\appendix

\section{Notation and preliminaries}
\label{app:preliminaries}

 We describe the experimental randomness by a probability space
$(\Omega,\mathcal{F},\mathbb{P})$, where $\Omega$ is the set of
possible realizations, $\mathcal{F}$ is a sigma-algebra of events,
and $\mathbb{P}$ assigns their probabilities.
A filtration $\mathcal{F}_0\subseteq\mathcal{F}_1\subseteq\cdots
\subseteq\mathcal{F}$ is an increasing sequence of sigma-algebras
representing the information available as the experiment
progresses~\cite{williams1991probability}.

Throughout the appendices, $\mathcal{F}_{t-1}$ denotes the full
experimental history before the measurement in round $t$, including source variables
that may be inaccessible to the experimenter.
The state $\rho_t$, observable $O_t$, and real reconstruction
coefficients $x_{t,k}$ are $\mathcal{F}_{t-1}$-measurable.
Such sequences are called predictable.
Given a fixed POVM $\{E_k\}$,
\begin{align}
 \sum_kx_{t,k}E_k=O_t, \quad
 \mathbb{P}(k_t=k\mid\mathcal{F}_{t-1})
 =\operatorname{tr}(\rho_t E_k ).
\end{align}
When using randomized measurements, the setting probability is included
in $E_k$, and the current setting is drawn independently of the past.
For classical shadows, a predictable dual frame
$\{\hat{\rho}_{t,k}\}_k$ gives the reconstruction coefficients
$x_{t,k}=\operatorname{tr}(\hat{\rho}_{t,k} O_t)$.

For $X_t=x_{t,k_t}$, the reconstruction identity gives
\begin{align}
\mathbb{E}[X_t\mid\mathcal{F}_{t-1}]
 =\operatorname{tr}(\rho_t O_t).
\end{align}
We define the trajectory average and the sample mean of the single-shot estimators by
\begin{align}
 \bar{o}_N=\frac{1}{N}\sum_{t=1}^N \operatorname{tr}(\rho_t O_t),
 \qquad
 \bar{X}_N=\frac{1}{N}\sum_{t=1}^NX_t.
\end{align}
The soft-truncated mean is
\begin{align}
 \bar{Y}_N
 =\frac{1}{\alpha N}\sum_{t=1}^N\psi(\alpha X_t),
\end{align}
where $\alpha>0$ and
$\psi(u)=\operatorname{sgn}(u)\log(1+|u|+u^2/2)$.

An integrable process $M_t$ is a martingale with respect to
$\{\mathcal{F}_t\}$ if $M_t$ is $\mathcal{F}_t$-measurable and
\begin{align}
 \mathbb{E}[M_t\mid\mathcal{F}_{t-1}]=M_{t-1}.
\end{align}
Thus its next increment has mean zero given the past.
Under the bounds used below, the accumulated error
\begin{align}
 M_N=\sum_{t=1}^N(X_t-\operatorname{tr}(\rho_t O_t))
 =N(\bar{X}_N-\bar{o}_N)
 \label{eq:error_martingale}
\end{align}
is a martingale starting at zero.

A supermartingale satisfies $ \mathbb{E}[M_t\mid\mathcal{F}_{t-1}]\leq M_{t-1}$.
For a nonnegative supermartingale $Z_N$ with $Z_0=1$, Ville's
inequality states that, for every $\delta\in(0,1)$,
\begin{align}
 \mathbb{P}\left(\sup_{N\geq0}Z_N\geq\frac1\delta\right)
 \leq\delta.
 \label{eq:ville}
\end{align}
It bounds the probability of ever crossing the threshold $1/\delta$,
without fixing a final sample size~\cite{howard2021timeuniform}.
This is the step that turns the exponential supermartingales
constructed below into confidence bounds valid simultaneously
at all sample sizes.

\section{Proof of Theorem~\ref{thm:catoni}}
\label{app:catoni}

 We prove Theorem~\ref{thm:catoni} using Catoni's
soft-truncation argument~\cite{catoni2012challenging}.
This is the conditional-second-moment argument of
Bhatt et al.~\cite[Lemma~6]{bhatt2022offline}.
Related Catoni supermartingales and extensions to changing
conditional means are discussed in Ref.~\cite{wang2023catoni}.

\begin{theoremrestate}{thm:catoni}
Suppose that, for every $t$,
\begin{align}
\mathbb{E}[X_t\mid\mathcal{F}_{t-1}]=\operatorname{tr}(\rho_t O_t), \qquad
 \mathbb{E}[X_t^2\mid\mathcal{F}_{t-1}]\leq s^2,
\end{align}
with a fixed constant $s>0$.
For a sample size $N$ fixed in advance and
$\delta\in(0,1)$, choose
$\alpha=\sqrt{2\log(2/\delta)/(Ns^2)}$.
Then, with probability at least $1-\delta$,
\begin{align}
 \left|\bar{Y}_N-\bar{o}_N\right|
 \leq\sqrt{\frac{2s^2\log(2/\delta)}{N}}.
\end{align}
\end{theoremrestate}

\noindent\textit{Proof.}
The key step is to turn the conditional second-moment bound
into an exponential bound.
The soft-truncation function $\psi$ satisfies
\begin{align}
 -\log(1-u+u^2/2)
 \leq\psi(u)
 \leq\log(1+u+u^2/2).
 \label{eq:catoni-envelope}
\end{align}
For $u\geq0$, the upper inequality is an equality, and the lower bound follows from
\begin{align}
 &\log(1-u+u^2/2)+\log(1+u+u^2/2)
 \notag\\
 &\qquad=\log(1+u^4/4)\geq0.
\end{align}
For $u<0$, the lower bound is an equality by definition,
and the same identity proves the upper bound.

Let $q_t=\mathbb{E}[X_t^2\mid\mathcal{F}_{t-1}]\leq s^2$.
For any fixed $\alpha>0$, exponentiating the bounds in
Eq.~\eqref{eq:catoni-envelope} and taking conditional
expectations gives
\begin{align}
 \mathbb{E}[e^{\pm\psi(\alpha X_t)}\mid\mathcal{F}_{t-1}]
 &\leq1\pm\alpha \operatorname{tr}(\rho_t O_t)+\frac{\alpha^2q_t}{2}
 \notag\\
 &\leq\exp\left(\pm\alpha \operatorname{tr}(\rho_t O_t)+\frac{\alpha^2q_t}{2}\right).
\end{align}

Define
\begin{align}
 Z_N^\pm
 =\exp\left[
 \pm\sum_{t=1}^N\{\psi(\alpha X_t)-\alpha \operatorname{tr}(\rho_t O_t)\}
 -\frac{\alpha^2}{2}\sum_{t=1}^Nq_t
 \right],
\end{align}
where $Z_0^\pm=1$. Since $\operatorname{tr}(\rho_t O_t)$ and $q_t$ are predictable,
\begin{align}
 \mathbb{E}[Z_t^\pm\mid\mathcal{F}_{t-1}]
 &=
 Z_{t-1}^\pm e^{\mp\alpha \operatorname{tr}(\rho_t O_t)-\alpha^2q_t/2}
 \mathbb{E}[e^{\pm\psi(\alpha X_t)}
 \mid\mathcal{F}_{t-1}]
 \notag\\
 &\leq Z_{t-1}^\pm.
\end{align}
Thus $Z_N^\pm$ are nonnegative supermartingales, and
iterated conditioning gives $\mathbb{E}[Z_N^\pm]\leq1$.

At a sample size $N$ fixed in advance, Markov's inequality gives
$\mathbb{P}(Z_N^\pm>2/\delta)\leq\delta/2$ for each sign.
A union bound therefore ensures that $Z_N^\pm\leq2/\delta$
holds for both signs with probability at least $1-\delta$.
On this event, taking logarithms and using the definitions of
$\bar{Y}_N$ and $\bar{o}_N$ gives
\begin{align}
 \pm\alpha N(\bar{Y}_N-\bar{o}_N)
 \leq\frac{\alpha^2}{2}\sum_{t=1}^Nq_t
 +\log\frac{2}{\delta}.
\end{align}
Using $\sum_{t=1}^Nq_t\leq Ns^2$ and dividing by $\alpha N$ yields
\begin{align}
 |\bar{Y}_N-\bar{o}_N|
 \leq\frac{\alpha s^2}{2}
 +\frac{\log(2/\delta)}{\alpha N}.
 \label{eq:catoni-alpha-bound}
\end{align}
The prescribed choice of $\alpha$ balances the two terms and
gives the bound in Theorem~\ref{thm:catoni}.

\section{Proof of Theorem~\ref{thm:eb}}
\label{app:eb}

 We prove Theorem~\ref{thm:eb} using the
construction of Howard et al.~\cite{howard2021timeuniform}.
Recall the observed sum of squared residuals,
\begin{align}
 S_N=\sum_{t=1}^N(X_t-c_t)^2.
\end{align}

\begingroup
\renewcommand{\thereptheorem}{\ref{thm:eb}}
\begin{reptheorem}
Let $c_t$ be predictable and suppose that, for every $t$,
\begin{align}
\mathbb{E}[X_t\mid\mathcal{F}_{t-1}]=\operatorname{tr}(\rho_t O_t),
 \qquad |X_t-c_t|\leq B,
\end{align}
with a fixed constant $B>0$.
Let $\phi(\lambda)=-\log(1-\lambda)-\lambda$.
For $\delta\in(0,1)$, define
\begin{align}
 \epsilon_N
 =\min_{0\leq j\leq J_N}
 \frac{B^2\log (2/\delta a_j)
       +\phi(\lambda_j)S_N}
      {N\lambda_j B},
\end{align}
where $J_N=\lceil\frac12\log_2N\rceil$ and, for $j=0,1,\ldots$,
\begin{align}
 a_j=[(j+1)(j+2)]^{-1},
 \qquad \lambda_j=2^{-(j+1)}.
 \label{eq:grid}
\end{align}
Then, with probability at least $1-\delta$,
\begin{align}
 |\bar X_N-\bar o_N|\leq\epsilon_N
\end{align}
for every $N\geq1$.
\end{reptheorem}
\endgroup

\noindent\textit{Proof.}
We construct a nonnegative supermartingale that relates
the accumulated error to the observed squared residuals.
Ville's inequality then gives a bound valid at every
sample size.

The empirical radius uses the observed residual $X_t-c_t$, but we want to bound the estimation error $X_t-\operatorname{tr}(\rho_t O_t)$. Define $R_t=(X_t-c_t)/B$ and $m_t=\mathbb{E}[R_t\mid\mathcal{F}_{t-1}]=(\operatorname{tr}(\rho_t O_t)-c_t)/B$. Then $|R_t|\leq1$ and
\begin{align}
 B(R_t-m_t)=X_t-\operatorname{tr}(\rho_t O_t).
\end{align}
Thus subtracting the conditional mean of the residual
recovers the estimation error, regardless of the accuracy
of the prediction $c_t$.

For $|u|\leq1$ and $0<\lambda<1$, expanding the logarithm gives
\begin{align}
 \lambda u-\log(1+\lambda u) \leq u^2\sum_{r=2}^{\infty}\frac{\lambda^r}{r}
 =\phi(\lambda)u^2.
\end{align}
Consequently,
\begin{align}
 e^{\lambda u-\phi(\lambda)u^2}
 \leq1+\lambda u.
 \label{eq:scalar-eb}
\end{align}
Apply this inequality to $R_t$ and $-R_t$, multiply by
$e^{\mp\lambda m_t}$, and take conditional expectations.
Since $m_t$ is determined by the past, this gives
\begin{align}
 \mathbb{E}\left[
 e^{\pm\lambda(R_t-m_t)-\phi(\lambda)R_t^2}
 \mid\mathcal{F}_{t-1}\right] &\leq
 e^{\mp\lambda m_t}(1\pm\lambda m_t) \nonumber\\
 &\leq1.
\end{align}
Define
\begin{align}
 Z_N^\pm(\lambda)
 =\exp\left[
 \pm\frac{\lambda M_N}{B}
 -\frac{\phi(\lambda)S_N}{B^2}
 \right],
 \label{eq:eb-supermartingale}
\end{align}
where $M_N=\sum_{t=1}^N(X_t-\operatorname{tr}(\rho_t O_t))$ is the accumulated error
and $Z_0^\pm(\lambda)=1$. The update from round $t-1$ to round $t$ is
\begin{align*}
 Z_t^\pm(\lambda)
 =Z_{t-1}^\pm(\lambda)
 e^{\pm\lambda(R_t-m_t)-\phi(\lambda)R_t^2}.
\end{align*}
Since $Z_{t-1}^\pm(\lambda)$ is determined by the past,
the preceding conditional-expectation bound gives
\begin{align*}
 \mathbb{E}[Z_t^\pm(\lambda)\mid\mathcal{F}_{t-1}]
 \leq Z_{t-1}^\pm(\lambda).
\end{align*}
Thus $Z_N^\pm(\lambda)$ are nonnegative supermartingales.

Applying Ville's inequality at
threshold $2/(\delta a_j)$ and taking a union bound over
both signs and all $j$, with $\sum_{j\geq0}a_j=1$, gives
\begin{align}
 |\bar X_N-\bar o_N|
 \leq
 \frac{B^2\log[2/(\delta a_j)]
       +\phi(\lambda_j)S_N}
      {N\lambda_jB}
\end{align}
simultaneously for every $N\geq1$ and $j\geq0$, with
probability at least $1-\delta$.
Since this event controls every $N$ and $j$, minimizing
over $0\leq j\leq J_N$ proves the theorem.

\subsection{Explicit bound on the confidence radius}

 We derive Eq.~\eqref{eq:empirical-width} by choosing a grid point
in Eq.~\eqref{eq:empirical-radius}. For each $N\geq1$, set
\begin{align}
 J_N &= \left\lceil \frac{1}{2}\log_2 N \right\rceil,
 \notag\\
 \ell_N &= \log\frac{2(J_N+1)(J_N+2)}{\delta}.
 \label{eq:grid-log}
\end{align}
Write $z=S_N/B^2$ and $L=\ell_N$. For $0\leq j\leq J_N$, we have
$\log[2/(\delta a_j)]\leq L$. Moreover,
$\phi(\lambda)\leq\lambda^2/[2(1-\lambda)]\leq\lambda^2$
for $0<\lambda\leq1/2$. Thus each of these grid points gives
\begin{align}
 \frac{N\epsilon_N}{B}
 \leq \frac{L}{\lambda_j}+\lambda_j z.
\end{align}

Choose the smallest $j$ for which $\lambda_j^2 z\leq L$.
Such a point exists because $z\leq N$ and
$\lambda_{J_N}\leq 1/(2\sqrt N)$, so
$\lambda_{J_N}^2 z\leq 1/4<L$.

If $j=0$, then $\lambda_0=1/2$ and $z\leq4L$, giving
$N\epsilon_N/B\leq 2L+z/2\leq4L$.
If $j\geq1$, the preceding grid point
$\lambda_{j-1}=2\lambda_j$ failed the selection condition. Hence
\begin{align}
 \lambda_j^2 z\leq L<4\lambda_j^2 z.
\end{align}
The left inequality gives $\lambda_j z\leq\sqrt{Lz}$,
and the right gives $L/\lambda_j<2\sqrt{Lz}$.
Thus $N\epsilon_N/B\leq3\sqrt{Lz}$ in this case.

Combining the two cases and substituting $z=S_N/B^2$ and
$L=\ell_N$, we obtain
\begin{align}
 \epsilon_N
 &\leq \frac{1}{N}
 \max\!\left\{4B\ell_N,\,3\sqrt{S_N\ell_N}\right\}
 \notag\\
 &\leq \frac{3\sqrt{S_N\ell_N}}{N}
      +\frac{4B\ell_N}{N}.
 \label{eq:explicit-grid}
\end{align}

\subsection{Stopping}

 Theorem~\ref{thm:eb} bounds the estimation error simultaneously
for all $N$. Hence, for any stopping time $\tau$,
\begin{align}
 \mathbb{P}\!\left(
 \tau<\infty,\,
 |\bar X_\tau-\bar o_\tau|>\epsilon_\tau
 \right)\leq\delta.
 \label{eq:general-stopping}
\end{align}
For the precision rule in Eq.~\eqref{eq:stopping-rule},
$S_N\leq B^2N$ and Eq.~\eqref{eq:explicit-grid} give
\begin{align}
 \epsilon_N
 \leq 3B\sqrt{\frac{\ell_N}{N}}+\frac{4B\ell_N}{N}
 \longrightarrow 0.
\end{align}
Thus, if measurements continue, the radius eventually falls below
any prescribed $\epsilon>0$. At $\tau_\epsilon$ the radius is at
most $\epsilon$, so Eq.~\eqref{eq:general-stopping} proves
Corollary~\ref{cor:stopping}. A witness threshold may never be
reached, but whenever it is reached, the same confidence event
certifies the entangled fraction up to that round.

\section{Learning the reconstruction and prediction}
\label{app:learning}

Before processing the data from round $t$, we use only outcomes
from earlier rounds to construct a state estimate
$\sigma_t$ and predict the probability of each outcome,
$\hat p_{t,k}=\operatorname{tr}(\sigma_t E_k)$. We also choose
a prediction $c_t$ for the expectation value of $O_t$. One choice
is $c_t=\operatorname{tr}(\sigma_t O_t)$, but any prediction based
on earlier data is admissible.

For an overcomplete measurement, different reconstruction coefficients
can represent the same observable~\cite{innocenti2023shadow,caprotti2024optimizing,fischer2024dual,malmi2024enhanced,korhonen2025improving}. We use this freedom to
minimize the expected squared residual under $\sigma_t$. Once $c_t$
has been chosen, every exact reconstruction $\sum_k x_kE_k=O_t$
satisfies
\begin{align}
 \sum_k\hat p_{t,k}(x_k-c_t)^2
 \!=\! \sum_k\hat p_{t,k}x_k^2
   -2c_t\operatorname{tr}(\sigma_t O_t)+c_t^2.
\end{align}
Only the first term on the right depends on the coefficients $x_k$.
This is why the objective in Eq.~\eqref{eq:online-dual} is the
predicted second moment $\sum_k\hat p_{t,k}x_k^2$.

The two constraints in Eq.~\eqref{eq:online-dual} serve different
purposes. The first, $\sum_k x_kE_k=O_t$, ensures that the
single-shot estimator is conditionally unbiased for the actual
preparation $\rho_t$, even when $\sigma_t$ is inaccurate:
\begin{align}
 \mathbb{E}[X_t\mid\mathcal F_{t-1}]
 &=\sum_k x_k\operatorname{tr}(\rho_t E_k)
 =\operatorname{tr}(\rho_t O_t).
\end{align}
The second, $|x_k-c_t|\leq B$ for every $k$, ensures that every
possible outcome satisfies the residual bound required by
Theorem~\ref{thm:eb}.

The resulting optimization is a convex quadratic program. To
guarantee a feasible solution, we choose $B$ large enough to
admit the standard reconstruction. For example, if
$|x_{t,k}^{\mathrm{std}}|\leq A$ and $|c_t|\leq C$, then
$|x_{t,k}^{\mathrm{std}}-c_t|\leq A+C$, so $B=A+C$ suffices.

Solving a full quadratic program can become costly for many-qubit systems. For many-qubit product measurements, one can reduce the cost of Eq.~\eqref{eq:online-dual} by constructing dual frames on individual qubits or small groups of qubits from their marginal measurement statistics~\cite{fischer2024dual,korhonen2025improving}. The tensor product of these local duals gives an exact reconstruction, although restricting the dual to this form can increase the variance. Our framework accommodates these constructions, or any other dual frame chosen using only earlier outcomes, provided the resulting coefficients satisfy the uniform residual bound
$|x_{t,k}-c_t|\leq B$ required by Theorem~\ref{thm:eb}.

\section{Clifford and Pauli second moments}
\label{app:shadow-moments}

To apply Theorem~\ref{thm:catoni}, we need a state-independent
bound on the second moment of the single-shot estimator
$X=x_k$ of $\operatorname{tr}(\rho O)$.
For an exact reconstruction $\sum_kx_kE_k=O$ with real
coefficients $x_k$, the worst-case second moment is
\begin{align}
\!\!\sup_\rho\mathbb{E}_\rho[X^2]\!=\sup_\rho\operatorname{tr}\left(\rho\sum_kx_k^2E_k\right)\!=\!\left\lVert\sum_kx_k^2E_k\right\rVert_\infty\!,\!
\end{align}
where $\mathbb{E}_\rho$ averages over the measurement settings
and outcomes, and the supremum is over density operators.
For a fixed dual frame, $x_k=\operatorname{tr}(\hat\rho_k O)$,
this worst-case second moment defines the squared shadow norm
$\lVert O\rVert_{\mathrm S}^2$.
We now evaluate these moments for the standard Clifford
and Pauli reconstructions.

For global Clifford measurements on $n$ qubits, let $d=2^n$.
For traceless Hermitian $O$, the standard single-shot estimator is
$X=(d+1)\operatorname{tr}(U^\dagger|z\rangle\langle z|U \, O)$,
where $U$ is a uniformly random Clifford unitary and $z$ is the
computational-basis measurement outcome.
The Clifford $3$-design identity gives~\cite{huang2020predicting}
\begin{align}
 \mathbb{E}_\rho[X^2]=
 \frac{d+1}{d+2}
 \left[
 \operatorname{tr}(O^2)+2\operatorname{tr}(\rho O^2)
 \right]\leq3\operatorname{tr}(O^2).
 \label{eq:clifford-second-moment}
\end{align}
Maximizing over $\rho$ therefore gives
$\lVert O\rVert_{\mathrm{S}}^2\leq3\operatorname{tr}(O^2)$.

For local Pauli measurements, the basis
$B_i\in\{X,Y,Z\}$ on each qubit is chosen independently
and uniformly.
If the outcome is $r_i=\pm1$, the standard local Pauli snapshot is
\begin{align}
 \hat{\rho}^{(i)}
 =\frac{\mathds{1}+3r_iB_i}{2}.
\end{align}
The full snapshot is $\hat{\rho}=\bigotimes_i\hat{\rho}^{(i)}$,
and each local factor has trace one.
For a Pauli operator $A\in\{X,Y,Z\}$,
Pauli orthogonality gives
$\operatorname{tr}(\hat{\rho}^{(i)} A)=3r_i$ if $A=B_i$
and zero otherwise.

Now let $P_j$ be a Pauli string with support $S_j$
and weight $w_j=|S_j|$.
A setting matches $P_j$ when $B_i=P_j^{(i)}$ for every
$i\in S_j$. The single-shot estimator $X_{P_j}=\operatorname{tr}(\hat{\rho} P_j)$
estimates the expectation value $\operatorname{tr}(\rho P_j)$.
The tensor-product structure of the snapshot gives
\begin{align}
 X_{P_j}
 &=\prod_i\operatorname{tr}
 (\hat{\rho}^{(i)} P_j^{(i)})
 \notag\\
 &=
 \begin{cases}
  3^{w_j}s_j, & \text{if the setting matches }P_j,\\
  0, & \text{otherwise},
 \end{cases}
\end{align}
where $s_j=\prod_{i\in S_j}r_i$ is the measured parity.
Each required basis matches with probability $1/3$,
so a full match occurs with probability $3^{-w_j}$.
The squared estimator is $3^{2w_j}$ on a match and zero
otherwise, giving
\begin{align}
 \mathbb{E}_\rho[X_{P_j}^2]
 =3^{-w_j}3^{2w_j}=3^{w_j}
\end{align}
for every state.

For $H=\sum_j\alpha_jP_j$ with real coefficients $\alpha_j$,
the single-shot estimator
$X_H=\operatorname{tr}(\hat{\rho} H)
=\sum_j\alpha_jX_{P_j}$
estimates $\operatorname{tr}(\rho H)$.
To calculate its second moment, we also need the cross moments
$\mathbb{E}_\rho[X_{P_j}X_{P_l}]$, since all string estimates
use the same measurement data.
Write $j\sim l$ when $P_j$ and $P_l$ agree on their
overlapping supports, and call such strings compatible.
If they disagree on any shared qubit, no setting can match
both strings, so $X_{P_j}X_{P_l}=0$.

Compatible strings match jointly with probability
$3^{-|S_j\cup S_l|}$.
On a joint match, shared outcomes appear twice in $s_js_l$
and drop out because $r_i^2=1$.
For example, the parities of $Z_1Z_2$ and $Z_2Z_3$ multiply as
$(r_1r_2)(r_2r_3)=r_1r_3$, which is the outcome for $Z_1Z_3$.

The same reasoning applies to any compatible pair:
$s_js_l$ is the outcome of measuring $P_jP_l$.
Its mean on a joint match is therefore
$\operatorname{tr}(\rho P_jP_l)$ by the Born rule.
Since $X_{P_j}X_{P_l}=3^{w_j+w_l}s_js_l$ on a joint match
and vanishes otherwise, we obtain
\begin{align}
 \mathbb{E}_\rho[X_{P_j}X_{P_l}]
 &=
 3^{w_j+w_l-|S_j\cup S_l|}
 \operatorname{tr}(\rho P_jP_l)
 \notag\\
 &=
 3^{|S_j\cap S_l|}
 \operatorname{tr}(\rho P_jP_l),
 \qquad j\sim l.
\end{align}
Expanding $X_H^2$ and using
$|\operatorname{tr}(\rho P_jP_l)|\leq1$ gives~\cite{hadfield2022measurements}
\begin{align}
 \mathbb{E}_\rho[X_H^2]
 &=
 \sum_{j\sim l}\alpha_j\alpha_l
 3^{|S_j\cap S_l|}
 \operatorname{tr}(\rho P_jP_l)
 \notag\\
 &\leq
 \sum_{j\sim l}|\alpha_j\alpha_l|
 3^{|S_j\cap S_l|}
 =V_H.
\end{align}
The sum includes both $(j,l)$ and $(l,j)$, as well as $j=l$.
Maximizing over $\rho$ yields
$\lVert H\rVert_{\mathrm{S}}^2\leq V_H$.

\section{Entanglement along a trajectory}
\label{app:witness}

\subsection{Entangled fraction and average robustness}

 Let $W_t$ be an entanglement witness chosen before processing round $t$.
If $W_t\geq-\mathds{1}$, then $-\operatorname{tr}(\rho_t W_t)$
is at most one for any preparation and at most zero for a separable
preparation. Hence
\begin{align}
 -\operatorname{tr}(\rho_t W_t)
 \leq \mathbf{1}\{\rho_t\text{ is entangled}\}.
\end{align}
Let $L_N$ be a lower confidence bound on the average negative
witness expectation $g_N=-\frac{1}{N}\sum_{t=1}^N\operatorname{tr}(\rho_t W_t)$.
On the event $g_N\geq L_N$, averaging the inequality above gives
\begin{align}
 f_N=\frac{1}{N}\sum_{t=1}^N
 \mathbf{1}\{\rho_t\text{ is entangled}\}
 \geq \max\{0,L_N\}.
\end{align}
Theorem~\ref{thm:catoni} gives this certificate at a prescribed sample size. Theorem~\ref{thm:eb} gives it simultaneously for all $N$, including at any finite stopping time.

If each witness also satisfies $W_t\leq\mathds{1}$, the same
data bound the average generalized robustness
\cite{brandao2005quantifying,eisert2007quantitative}.
The generalized robustness $R_{\mathrm g}(\rho)$ is the smallest
$r\geq0$ such that $(\rho+r\tau)/(1+r)$ is separable for some state
$\tau$~\cite{steiner2003generalized}. For any such mixture,
\begin{align}
 0
 &\leq \operatorname{tr}(\rho W_t)+r\operatorname{tr}(\tau W_t)
 \leq \operatorname{tr}(\rho W_t)+r.
\end{align}
The first inequality follows from separability and the second from
$W_t\leq\mathds{1}$. Thus
$R_{\mathrm g}(\rho)\geq-\operatorname{tr}(\rho W_t)$. Averaging
over the preparations gives, on the same confidence event,
\begin{align}
 \frac{1}{N}\sum_{t=1}^N R_{\mathrm g}(\rho_t)
 \geq \max\{0,L_N\}.
 \label{eq:robustness}
\end{align}

\subsection{The drifting Bell source}

Consider the source and witness of Sec.~\ref{sec:witness},
\begin{align}
 \rho_t
 &=v|\Phi_{\phi_t}\rangle\langle\Phi_{\phi_t}|
   +(1-v)\frac{\mathds{1}}{4},
 \notag\\
 W_t
 &=\mathds{1}-2|\Phi_{\theta_t}\rangle
   \langle\Phi_{\theta_t}|,
\end{align}
where $|\Phi_\phi\rangle=\frac{1}{\sqrt{2}}(|00\rangle+e^{i\phi}|11\rangle)$ and
$0\leq v\leq1$. For every separable state $\sigma$,
$\langle\Phi_{\theta_t}|\sigma|\Phi_{\theta_t}\rangle\leq1/2$, so
$W_t$ is an entanglement witness. Its eigenvalues are $-1$
and $1$, so the same lower bound $\max\{0,L_N\}$ applies to
the entangled fraction and the average generalized robustness.

For a witness phase $\theta_t$ chosen from earlier data, the
negative witness expectation is
\begin{align}
 -\operatorname{tr}(\rho_t W_t)
 =v\cos(\phi_t-\theta_t)-\frac{1-v}{2}.
 \label{eq:bell-gain}
\end{align}
A phase-matched witness gives $(3v-1)/2$, which is positive for
$v>1/3$. Thus every preparation is entangled in this regime.
For equally spaced phases $\phi_t=2\pi(t-1)/N$ with $N\geq2$,
however, $\sum_{t=1}^N e^{i\phi_t}=0$. The off-diagonal terms
cancel in the time-averaged state $\bar\rho_N$, leaving a separable
state diagonal in the computational basis.

To obtain the fixed-sample certificate, we use the Pauli expansion
\begin{align}
 \!-W_t
 =&-\frac{\mathds{1}}{2}+\frac{1}{2}ZZ
 \notag\\
 &+\frac{\cos\theta_t}{2}(XX-YY)
  +\frac{\sin\theta_t}{2}(XY+YX).\!
 \label{eq:bell-pauli}
\end{align}
We estimate the traceless part $O_t=-W_t+\mathds{1}/2$.
Under uniform local Pauli measurements, its five Pauli strings
correspond to distinct settings. A string with coefficient $a$
contributes a term of magnitude $9|a|$ to the standard single-shot estimator when its setting
is drawn, which occurs with probability $1/9$, and zero otherwise.
Since at most one string contributes in each round, the standard
single-shot estimator $X_t$ of $O_{t}$ satisfies
\begin{align}
 \mathbb{E}[X_t^2\mid\mathcal{F}_{t-1}]
 &=\frac{9}{4}
 \bigl(1+2\cos^2\theta_t+2\sin^2\theta_t\bigr)
 =\frac{27}{4}.
\end{align}
Theorem~\ref{thm:catoni} therefore applies with $s^2=27/4$.
Subtracting the known offset $1/2$ from the resulting lower
confidence bound gives the fixed-sample lower bound $L_N$ for the average negative witness
expectation.

\section{Hamiltonians and energy-density precision}
\label{app:manybody}

We derive precision bounds for energy-density estimates from uniform local Pauli measurements. For local Hamiltonians with bounded interaction strengths and a bounded number of terms per qubit, the variance added by the measurement
protocol grows at most linearly with system size.  We use this result to bound the empirical confidence radius in terms of system size, physical energy variance, and prediction accuracy. 

\subsection{Variance of the single-shot estimator}

 Consider an $n$-qubit Hamiltonian
$H=\sum_{j=1}^m\alpha_jP_j$, where each Pauli string acts
nontrivially on at most $k$ qubits, each qubit belongs to at most
$\ell$ terms, and $|\alpha_j|\leq J$. Let
$X=\operatorname{tr}(\hat\rho H)$ be its standard single-shot estimator under
uniform local Pauli measurements. We prove the bound stated in Eq.~\eqref{eq:measurement-variance}:
\begin{align}
 \operatorname{Var}_\rho(X)
 \leq \operatorname{Var}_\rho(H)+3^k k\ell mJ^2.
\end{align}
The first term is the physical energy variance. The second bounds the additional
variance introduced by the randomized measurement protocol.

Let $S_j$ be the support of $P_j$. We write $j\sim l$ when $P_j$
and $P_l$ have the same Pauli operator on every qubit in their
common support. This includes pairs with disjoint supports.
The moment formula of App.~\ref{app:shadow-moments}
gives~\cite{hadfield2022measurements}
\begin{align}
 \mathbb E_\rho[X^2]
 =\sum_{j\sim l}\alpha_j\alpha_l
 3^{|S_j\cap S_l|}\operatorname{tr}(\rho P_jP_l).
 \label{eq:pauli-moment}
\end{align}
The sum includes both $(j,l)$ and $(l,j)$, as well as $j=l$.
Since $\mathbb E_\rho[X]=\operatorname{tr}(\rho H)$, the squared
means cancel when we subtract the variances:
\begin{align}
 \operatorname{Var}_\rho(X)-\operatorname{Var}_\rho(H)
 ={}&\sum_{j\sim l}\alpha_j\alpha_l
 (3^{|S_j\cap S_l|}-1)\operatorname{tr}(\rho P_jP_l)
 \notag\\
 &-\sum_{j\not\sim l}\alpha_j\alpha_l
 \operatorname{tr}(\rho P_jP_l).
 \label{eq:variance-decomposition}
\end{align}

Only overlapping pairs contribute to this difference. Indeed,
disjoint strings are compatible and have
$3^{|S_j\cap S_l|}-1=0$. For a compatible overlapping pair,
the contribution has absolute value at most $(3^k-1)J^2$.
For an incompatible pair, it is at most $J^2$, since
$|\operatorname{tr}(\rho P_jP_l)|\leq1$. We therefore bound
either contribution by $3^kJ^2$.

For each fixed $j$, the support $S_j$ contains at most $k$
qubits, and each qubit belongs to at most $\ell$ terms.
Hence at most $k\ell$ terms overlap $S_j$, including $j$ itself.
There are therefore at most $mk\ell$ contributing ordered pairs
across the two sums. Bounding their contributions in absolute
value proves Eq.~\eqref{eq:measurement-variance}.

For the magnetization $H=\sum_{i=1}^n Z_i$, distinct terms have
disjoint supports, so only the $n$ diagonal pairs contribute.
Each contributes $3-1=2$, giving
\begin{align}
 \operatorname{Var}_\rho(X)
 =\operatorname{Var}_\rho(H)+2n.
\end{align}
The corresponding second-moment operator is
$H^2+2n\mathds{1}$. Its largest eigenvalue is $n^2+2n$, yielding
$\lVert H\rVert_{\mathrm S}^2=n^2+2n$, as stated in
Eq.~\eqref{eq:magnetization}.

To obtain an overall variance of order $n$, we also need to control
the physical energy fluctuations. Write $h_j=\alpha_jP_j$. Since the energy variance is the sum of
the symmetrized covariances between these local terms, a sufficient
condition is
\begin{align}
 \sum_{j,l}\left|
 \tfrac12\operatorname{tr}(\rho_t\{h_j,h_l\})
 -\operatorname{tr}(\rho_t h_j)\operatorname{tr}(\rho_t h_l)
 \right|\leq C_0n,
 \label{eq:spatial-covariance-condition}
\end{align}
where $C_0$ is independent of $n$ and $t$.
For example, this condition holds if the sum of the absolute
covariances of each term with all other terms is bounded
independently of system size, since there are only $m=O(n)$ terms.
Removing the absolute values gives
$\operatorname{Var}_{\rho_t}(H)$, so the condition implies
$\operatorname{Var}_{\rho_t}(H)\leq C_0n$.

\subsection{Scaling of the energy-density precision}

We now bound the empirical confidence radius for the average
energy density defined in Eq.~\eqref{eq:energy-density-target}.
Before each measurement, we use earlier outcomes to choose
a prediction $c_t$ of $\operatorname{tr}(\rho_t H)$.
The single-shot estimator satisfies
$\mathbb E[X_t\mid\mathcal F_{t-1}]
=\operatorname{tr}(\rho_t H)$.
We assume that, for every $t$,
\begin{align}
 \operatorname{Var}(X_t\mid\mathcal F_{t-1})\leq Cn,
 \qquad |X_t-c_t|\leq B=bn,
 \label{eq:energy-score-assumptions}
\end{align}
where $C,b>0$ are independent of $n$ and $t$.
The variance bound follows from the preceding analysis.

By Eq.~\eqref{eq:explicit-grid}, the energy-density radius obeys
\begin{align}
 r_N\leq\frac{3\sqrt{S_N\ell_N}+4B\ell_N}{nN},
 \quad \:
 S_N=\sum_{t=1}^N(X_t-c_t)^2,
\end{align}
where $\ell_N$ is defined in Eq.~\eqref{eq:grid-log}.
It therefore remains to control the observed sum of squared
residuals $S_N$.
We first bound the corresponding sum of conditional expectations.
By Eq.~\eqref{eq:variance-prediction},
\begin{align}
 Q_N
 &=\sum_{t=1}^N
 \mathbb E[(X_t-c_t)^2\mid\mathcal F_{t-1}]
 \notag\\
 &=\sum_{t=1}^N
 \operatorname{Var}(X_t\mid\mathcal F_{t-1})+D_N
 \notag\\
 &\leq CnN+D_N,
\end{align}
where
$D_N=\sum_{t=1}^N(\operatorname{tr}(\rho_t H)-c_t)^2$
is the accumulated squared error in predicting the mean energy.

To relate $S_N$ to $Q_N$, we use
$e^y\leq1+(e-1)y$ for $0\leq y\leq1$, with
$y=(X_t-c_t)^2/B^2$.
Taking conditional expectations and using $1+u\leq e^u$
shows that
\begin{align}
 M_N=\exp\!\left[
 \frac{S_N-(e-1)Q_N}{B^2}
 \right]
\end{align}
is a nonnegative supermartingale with $M_0=1$.
For any $\eta\in(0,1)$, Ville's inequality therefore gives,
with probability at least $1-\eta$, simultaneously for all $N$,
\begin{align}
 S_N\leq(e-1)(CnN+D_N)+B^2\log(1/\eta).
 \label{eq:residual-upper}
\end{align}
With probability at least $1-\delta-\eta$, the confidence
intervals contain the average energy density and the bounds
below hold, simultaneously for all sample sizes.

Substituting Eq.~\eqref{eq:residual-upper} into the bound
on $r_N$ gives
\begin{align}
 r_N\lesssim
 \sqrt{\frac{C\mathcal L_N}{nN}}
 +\frac{\sqrt{D_N\mathcal L_N}}{nN}
 +\frac{b\mathcal L_N}{N},
 \label{eq:density-bound}
\end{align}
where $\mathcal L_N=\ell_N+\log(1/\eta)$ and
$\lesssim$ hides only a numerical constant.
For fixed $\delta$ and $\eta$,
$\mathcal L_N=O(\log\log N)$ as $N\to\infty$.
The three terms arise from the conditional estimator variance,
the accumulated prediction error, and the uniform residual
bound, respectively.
This proves Eq.~\eqref{eq:density-prediction-scaling}.

If, for all $N$,
$D_N\leq C_{\mathrm p}nN+C_{\mathrm{init}}n^2$,
with nonnegative constants independent of $n$ and $N$, then
\begin{align}
 r_N\lesssim
 \sqrt{\frac{(C+C_{\mathrm p})\mathcal L_N}{nN}}
 +\frac{(b+\sqrt{C_{\mathrm{init}}})\mathcal L_N}{N}.
 \label{eq:density-scaling-explicit}
\end{align}
This assumption allows an average squared prediction error
of order $n$, together with errors of magnitude $O(n)$
in a bounded number of initial rounds.
Up to logarithmic factors, the radius is therefore at most
of order $\frac{1}{\sqrt{nN}}+\frac{1}{N}$.

\section{Classical shadows for history-dependent quantum channels}
\label{app:channels}

 The Choi representation allows us to extend the state-estimation
results to quantum channels that depend on earlier experimental
outcomes. For a channel
$\Lambda:\mathcal L(\mathcal H)\to\mathcal L(\mathcal H)$,
the normalized Choi state is
\begin{align}
 \rho_\Lambda
 &=(\mathcal I\otimes\Lambda)(|\Omega\rangle\langle\Omega|),
\end{align}
where $ |\Omega\rangle=\frac{1}{\sqrt d}\sum_{i=1}^d|i\rangle\otimes|i\rangle$ and $d=\dim\mathcal H$.

In classical-shadow process tomography
\cite{kunjummen2023shadow,levy2024classical}, we prepare an
input state $\sigma_j$ chosen uniformly from an IC ensemble satisfying
$L^{-1}\sum_{j=1}^L\sigma_j=\frac{\mathds{1}}{d}$, apply the channel,
and measure the output using an IC POVM $\{E_k\}$.
The joint probability of the input $j$ and output
outcome $k$ is
\begin{align}
 p_{jk}
 &=\frac{1}{L}\operatorname{tr}[\Lambda(\sigma_j) E_k]
 =\operatorname{tr}(\rho_\Lambda F_{jk}),
 \label{eq:channel-povm}
\end{align}
where $ F_{jk}=\frac{d}{L}\sigma_j^T\otimes E_k$. The operators $\{F_{jk}\}$ form an IC POVM on
$\mathcal H\otimes\mathcal H$.
Thus the input--output experiment can be treated as a
measurement of the Choi state. Choose snapshots $\hat\rho_{jk}$ satisfying
$\sum_{j,k}\operatorname{tr}(\rho F_{jk})\hat\rho_{jk}=\rho$. Their expectation is therefore
$\mathbb E[\hat\rho_{jk}]=\rho_\Lambda$.

At round $t$, the channel $\Lambda_t$ may depend arbitrarily
on the previous history $\mathcal F_{t-1}$, but a single
channel must describe its response to every possible fresh
input. We choose the input label $j_t$ independently of the
past and prepare $\sigma_{j_t}$ without correlations with
the device memory.

The input--output experiment then satisfies the assumptions
of the state-shadow analysis, with $\rho_t$ replaced by
$\rho_{\Lambda_t}$. Theorems~\ref{thm:catoni} and~\ref{thm:eb},
including the corresponding stopping guarantees, therefore
apply to the trajectory average $\bar o_N=\frac{1}{N}\sum_{t=1}^N
 \operatorname{tr}(\rho_{\Lambda_t} O_t)$,
under their respective second-moment and residual bounds.
As before, the observables, reconstruction coefficients,
and predictions may be chosen from earlier data.

For example, choosing $O=d\,\sigma^T\otimes H$ gives
\begin{align}
 \operatorname{tr}(\rho_{\Lambda_t} O)
 =\operatorname{tr}[\Lambda_t(\sigma) H],
\end{align}
so the protocol estimates the average output expectation
for a specified input state $\sigma$ and observable $H$.

\section{Numerical implementation}
\label{app:numerics}
\label{app:mb-numerics}

We give the numerical details for the Bell-source and many-body
examples. All adaptive choices use only earlier outcomes.
The empirical bounds are updated after each processed outcome
using Eq.~\eqref{eq:empirical-radius}, with the specified finite
grid and the weights $a_j$.
Each method and target separately uses $\delta=0.05$.

\subsection{Bell-source examples}

We simulate the source of Eq.~\eqref{eq:bell-source}
with $v=0.9$ and $\phi_t=2\pi(t-1)/N$, where $N=32768$.
The phase completes one rotation every $N$ measurements.
We generate $50$ independent repetitions of $2N$ measurements,
covering two rotations in each repetition. In each round, both qubits are measured independently in uniformly random Pauli bases $a,b\in\{X,Y,Z\}$,
giving outcomes $s_1,s_2\in\{-1,1\}$.
Within each repetition, all methods use the same
measurement outcomes.
The true source parameters are not supplied to the analysis.

To follow the changing phase, we keep a separate average
of the outcome products $s_1s_2$ for each setting $ab$.
Let $m_{t,ab}$ denote this average using outcomes from
rounds before $t$, with all averages initialized at zero.
If round $t$ uses setting $ab$, we update
\begin{align}
 m_{t+1,ab}=0.99m_{t,ab}+0.01s_1s_2,
\end{align}
leaving the other averages unchanged.

For the Bell source, the $XY$ and $YX$ correlations
equal $v\sin\phi_t$, while the $XX$ and $YY$ correlations
equal $v\cos\phi_t$ and $-v\cos\phi_t$, respectively.
We therefore estimate the phase from
\begin{align}
 \theta_t=\operatorname{atan2}\left(
 \frac{m_{t,XY}+m_{t,YX}}{2},
 \frac{m_{t,XX}-m_{t,YY}}{2}
 \right),
\end{align}
and use $\theta_t$ as the adaptive witness phase.
We set $\theta_t=0$ when both arguments vanish.

Since $\operatorname{tr}(\rho_t\,Z\otimes Z)=v$,
we estimate $v$ by averaging the outcomes
from earlier $ZZ$ measurements.
We clip the estimate $\hat v_t$ to $[0,1]$
and set $\hat v_t=0$ before the first such measurement.
The estimated state and prediction are
\begin{align}
 \sigma_t
 =\hat v_t|\Phi_{\theta_t}\rangle\langle\Phi_{\theta_t}|
 +(1-\hat v_t)\frac{\mathds{1}}{4}, \quad
 c_t=-\operatorname{tr}(\sigma_t W_t),
\end{align}
where $W_t$ is the adaptive witness with phase $\theta_t$.

Writing
\begin{align}
 -W_t=-\frac{\mathds{1}}{2}
 +\sum_{a,b}\beta_{t,ab}\,\sigma_a\otimes\sigma_b,
\end{align}
the standard local Pauli estimate for setting $ab$ is
\begin{align}
 X_t^{\mathrm{std}}
 =-\frac12+9\beta_{t,ab}s_1s_2.
\end{align}
Using $\sigma_t$, we construct the exact reconstruction
\begin{align}
 X_t^{\mathrm{learn}}
 &=c_t+9\beta_{t,ab}[s_1s_2-\operatorname{tr}(\sigma_t\,\sigma_a\otimes\sigma_b)].
 \label{eq:bell-learned}
\end{align}
It remains unbiased for any prepared state and satisfies
$|X_t^{\mathrm{learn}}-c_t|\leq9$.

For the fixed-sample method, we apply
Theorem~\ref{thm:catoni} to the standard estimates of
$-W_t+\mathds{1}/2$, with $s^2=27/4$, and subtract
$1/2$ from the resulting lower confidence bound.
For the empirical method, we use the learned estimates
and Eq.~\eqref{eq:empirical-radius} with $B=9$.

We evaluate both lower bounds at $N=32768$.
For empirical stopping, we stop at the first round $t$
with $L_t\geq0.80$, allowing up to $2N$ measurements.

\subsection{States and simulated measurements}

For the Ising chain and system sizes specified in
Sec.~\ref{sec:numerical-simulations}, we compute ground states
using the density-matrix renormalization group
\cite{schollwock2011density}, implemented in TeNPy~1.1.1
\cite{hauschild2024tenpy}.
We use a bond dimension of at most $32$ and obtain
$\operatorname{Var}_{\rho}(H)<4\times10^{-9}$.
For magnetization, we use $|+\rangle^{\otimes n}$ and
the GHZ state
$(|0\rangle^{\otimes n}+|1\rangle^{\otimes n})/\sqrt2$.

In each round, we prepare the same state and measure
every qubit independently in a uniformly random Pauli basis,
sampling outcomes according to Born's rule.
Estimation errors are measured relative to the numerical
state used to generate the outcomes.

For each state and system size, we simulate $50$ independent
repetitions and apply all methods to the same outcomes.
We report medians, with error bars spanning the $25$th
to $75$th percentiles.
We compare precision after $2^{16}$ measurements.
For energy-density stopping, we check the empirical bound
after every round and stop when $r_N\leq\epsilon=0.01$.

\subsection{Estimation and confidence bounds}

For a total observable $G=\sum_j\alpha_jP_j$, the
standard local Pauli snapshot gives the estimator
\begin{align}
 X_t=\operatorname{tr}(\hat\rho_{t, k_t} G),
\end{align}
and the density estimator is $X_t/n$.
We predict $\operatorname{tr}(\rho_t G)$ using
\begin{align}
 c_t=\sum_j\alpha_j\bar m_{t,j},
\end{align}
where $\bar m_{t,j}$ is the average of earlier measurement
results for $P_j$. We use zero until $P_j$ has been
measured at least once.

We write $b$ for the residual bound per qubit, so that
$|X_t-c_t|/n \le b$.
For the Ising Hamiltonian, $|X_t|\leq9(n-1)$
and $|c_t|\leq2n-1$, so we use
\begin{align}
 b=11-\frac{10}{n}.
 \label{eq:mb-density-cap}
\end{align}
For magnetization, $|X_t|\leq3n$ and $|c_t|\leq n$,
giving $b=4$.

For density estimation, we apply
Eq.~\eqref{eq:empirical-radius} to $X_t/n$ and $c_t/n$,
using $b$.
Equivalently,
\begin{align}
 r_N=\frac{\epsilon_N}{n},
 \label{eq:mb-numerical-radius}
\end{align}
where $\epsilon_N$ is computed for the total observable
using the residual bound $B=nb$.

For the fixed-sample comparison, we use a state-independent
second-moment bound $s^2$ for the density estimator.
For the total Ising energy,
\begin{align}
 \mathbb E_\rho[X_t^2]
 =\operatorname{tr}(\rho H^2)+10n-8 +4\sum_{i=1}^{n-2}\operatorname{tr}(\rho Z_iZ_{i+2}).
\end{align}
The spectrum of $H$ is symmetric, so
$\lVert H\rVert_\infty=|E_0|$, where $E_0$ is the exact
ground energy. Bounding each correlation by one gives
\begin{align}
 s^2=\frac{E_0^2+14n-16}{n^2}.
\end{align}
For magnetization density, the exact worst-case second
moment is $s^2=1+2/n$.
Theorem~\ref{thm:catoni} then gives the fixed-sample
radius and sufficient measurement budget
\begin{align}
 r_N^{\mathrm{FS}}
 &=\sqrt{\frac{2s^2\log(2/\delta)}{N}},
 \\
 N_{\mathrm{FS}}
 &=\left\lceil
 \frac{2s^2\log(2/\delta)}{\epsilon^2}
 \right\rceil.
\end{align}

Since $\lambda_j\leq1/2$ and $a_j\leq1/2$,
Eq.~\eqref{eq:empirical-radius} gives
$r_N\geq2b\log(4/\delta)/N$,
even when the observed residuals vanish.
The empirical stopping rule therefore satisfies
\begin{align}
 \tau_\epsilon
 \geq\frac{2b\log(4/\delta)}{\epsilon}.
 \label{eq:mb-stopping-floor}
\end{align}

To assess the effect of prediction, we also compute the
empirical bound with $c_t=0$, keeping the same measurement
outcomes and $b$.
For the Ising chain at $n=128$ and $N=2^{16}$,
the median radius increases from $4.78\times10^{-3}$
to $1.89\times10^{-2}$.

\end{document}